\documentclass[11pt]{article}

\usepackage[utf8]{inputenc}
\usepackage[english]{babel}
\usepackage{microtype}

\usepackage[margin=1in]{geometry}

\usepackage{amsmath,amssymb,amsfonts,amsthm,bm,bbm}
\usepackage{mathtools}

\usepackage{graphicx}
\usepackage{float}
\usepackage{multirow}
\usepackage{subcaption}
\usepackage{algorithm}
\usepackage{algpseudocode}

\usepackage{enumitem}

\usepackage{natbib}
\usepackage[colorlinks=true,allcolors=blue]{hyperref}

\usepackage[font=small,labelfont=bf]{caption}

\theoremstyle{plain}
\newtheorem{theorem}{Theorem}
\newtheorem{lemma}{Lemma}

\theoremstyle{definition}

\theoremstyle{remark}

\newcommand{\iid}{\stackrel{\mbox{\scriptsize iid}}{\sim}}
\newcommand{\ind}{\stackrel{\mbox{\scriptsize ind}}{\sim}}

\usepackage{setspace}

\title{\textbf{Bayesian Mixture Models for Histograms:\\
with Applications to Large Datasets}}

\usepackage{authblk}

\author[1]{Richard L. Warr\thanks{warr@stat.byu.edu}}
\author[2]{Fernando A. Quintana}
\author[3]{Alessandra Guglielmi}
\author[4]{Mario Beraha}

\date{}

\affil[1]{Department of Statistics, Brigham Young University}
\affil[2]{Department of Statistics, Pontificia Universidad Católica de Chile}
\affil[3]{Department of Mathematics, Politecnico di Milano}
\affil[4]{Department of Economics, Management and Statistics, University of Milano-Bicocca}

\begin{document}

\maketitle

\vspace{-3em}

\begin{abstract}
In many real-world scenarios, especially those involving privacy constraints or data summarization, data are available only in aggregated forms, such as histograms or frequency tables. This work introduces a novel Bayesian method for inferring the underlying population distribution by fitting a mixture model to binned data. While we focus on mixtures of normal distributions, the framework is flexible and can be extended to other distributional families. We place a prior distribution on the number of mixture components, accommodating both finite and countably infinite mixtures, and perform inference using reversible jump MCMC. The proposed approach demonstrates strong performance on large-scale data, showcasing the potential of nonparametric Bayesian modeling in practical applications. Furthermore, we extend the method to model multiple histograms simultaneously and cluster them using the Dirichlet process. This enables information sharing across populations and provides a principled posterior probability to assess homogeneity between groups. Some theoretical results supporting the performance of our proposed methodology are also discussed.
\end{abstract}



\setstretch{1.2}

\bibliographystyle{apalike}

\section{Introduction}

In an era where data are often collected and shared in aggregated forms, such as histograms or frequency tables, there is a need for statistical methods that can accurately recover underlying population distributions from such summaries. This scenario often arises due to privacy constraints, storage limitations, or practical data collection processes. Traditional Bayesian modeling approaches often assume access to raw data, making them unsuitable or inefficient in these contexts.

To address this challenge, we propose a Bayesian mixture modeling approach that directly models histograms or binned data. Our model is a mixture distribution which estimates the latent probability density function of the underlying population. Although we focus primarily on a mixture of Gaussian distributions, the proposed framework is quite general and can be extended to other families of distributions. A prior is placed on the number of mixture components, allowing for both finite and countably infinite mixtures, and posterior samples are obtained using reversible jump Markov chain Monte Carlo (RJMCMC) methods. Furthermore, we extend our basic mixture model to handle multiple histograms by clustering them using a Dirichlet process (DP), which enables information sharing across groups and yields posterior probabilities for assessing homogeneity between populations.

While Bayesian nonparametric methods such as Dirichlet process mixture (DPM) models have shown strong theoretical appeal and modeling flexibility, their computational cost becomes prohibitive as the dataset grows. In particular, posterior sampling from DPM models scales poorly with the number of observations, making them impractical for modern large-scale data problems. A substantial body of research has explored alternatives and approximations to mitigate this issue, but many existing solutions involve complex or domain-specific trade-offs.

In contrast, the method we present is conceptually simple and computationally efficient, with inference that scales approximately linearly with the number of bins in the histogram or frequency table. Through simulation studies and theoretical developments, we demonstrate that our approach not only offers substantial speed improvements but also retains consistency with the true underlying distribution. These properties make our proposed model a promising and scalable solution for mixture modeling in the context of large or summarized datasets.

\subsection{Literature Review}

The need for statistical models for histograms or binned data has grown due to increasing concerns around privacy and large-scale data summarization. Traditional approaches, which rely on raw, individual-level data, are often unavailable or unsuitable in many contexts. We outline a few key works which led us to the development and/or facilitated the methodology discussed in this article. 

From a frequentist perspective, \cite{wasserman2006NONpar} demonstrated that classical histogram estimators can converge to the true density under mild smoothness conditions. Theoretical results are derived by introducing a histogram-based likelihood and proving convergence in the limit of growing sample sizes and bin counts. Alternatively, \cite{billard2017hierarchical} introduced hierarchical clustering techniques tailored for histogram-valued data, laying the foundation for symbolic data analysis in settings where only aggregated information is available. Furthermore, \cite{SameAmbroiseGovaert2006} consider a classification approach based on binned data, employing the EM algorithm. Similarly, \cite{Shirazi2023_EM_MCEM} use EM and Monte Carlo EM to estimate the underlying mean and variance structure from grouped (binned) data.

From a Bayesian point of view, 
\cite{LambertEilers2009} estimate smooth densities via penalized splines in the logarithmic scale, with a roughness penalty approach.
\cite{AlstonMengersen2010} consider density estimation for a finite Gaussian mixture model to estimate the underlying distribution from which binned data are collected. Parameter estimation under a similar framework is considered in \cite{GauTapsobaLee2014}.
More recently, and more closely connected to our proposal, \cite{martinez2024model} proposed a model-based approach for clustering binned data, demonstrating the advantages of probabilistic modeling frameworks for such summaries.  Additionally, \cite{simensen2026random} propose a Bayesian approach for constructing irregular histograms, in which both the number and locations of bins are learned from data. This complements our framework, which accommodates potentially unequal binning intervals and focuses on inference given aggregated observations.

The Reversible Jump MCMC algorithm introduced by \cite{richardson1997bayesian} allows users to estimate the number of components for mixture
models with a prior on the number of components. This
work explicitly allows implementation of posterior simulation algorithms for cases where
the proposed moves imply changes in the dimensions of the
underlying parameter vectors. 
Similarly, \cite{neal2000markov} outlined efficient MCMC methods for Dirichlet process mixture (DPM) models, which are widely used in nonparametric Bayesian inference. These computational approaches using one-at-a-time updates, while powerful, do not scale well as the number of observations becomes large. 

Not only might one be interested in estimating a distribution given by a single population, but clustering distributions of multiple populations could be beneficial.  Some previous works have approached this problem.  For example, the nested Dirichlet process \citep{rodriguez2008nested} allows for shared mixture distributions across groups.  The hierarchical Dirichlet process (HDP) framework introduced by \cite{teh2004sharing}, allows for shared mixture components across groups. Similar constructions are those in \cite{denti2023common} and \cite{dangelo2026BA}, based on the assumption that all groups share a common set of atoms and differ through group-specific weights.
A related yet more parsimonious approach is the semi-hierarchical Dirichlet process developed by \cite{beraha2021semi}, which provides computational and modeling advantages when the groups are expected to be partially but not fully exchangeable. 
Other works include \cite{Argiento2020}, \cite{lijoi2023flexible}, \cite{Duan2025} and references therein. A common aspect in these works is the construction of a model for several related distributions by relating atoms and/or weights related to individual observations and/or groups. In contrast, our proposal for a single histogram employs no atoms, while the extension to multiple histograms treats entire distributions as atoms. Again, this is a feature of the data, where only binned counts are available. 

Some generalizations of the notion of histograms have been explored in the literature. \cite{canale2011bayesian} developed mixtures of rounded continuous kernel distributions for discrete count data. More recently, \cite{goh2024sparse} developed sparse tree-based and list-based 
models for density estimation in the context of 
binary/categorical data.

\subsection{Our Contributions}

As our first contribution, we present a framework for analyzing data from a single histogram when the original observations are missing or masked for data privacy, and only binned frequencies are available. The goal is the estimation of the true underlying population distribution generating the original observations. As noted earlier, the great majority of nonparametric approaches rely on access to the original data, and thus, cannot be employed in
our context. A second contribution is an efficient posterior simulation scheme tailored to our first model that scales well with the number of bins, but that is not affected by the total sample size (i.e. the sum of bin frequencies). This is a genuinely
distinctive feature of our approach, compared to the usual
behavior reported for the computational aspects of traditional Bayesian nonparametric models. Our third contribution is the
extension of the model to the case of multiple samples
of binned data, where the parameters specifying each population distribution are assumed to be a sample from a Dirichlet process.
A fourth
contribution is a suitable posterior simulation scheme for
the multi-sample case, that employs some specifically
designed moves to facilitate the parameter space exploration.

We began with a discussion of previous work and the challenges and motivations for modeling histogram data.  The remainder of the paper is organized as follows. We introduce our modeling framework in Section~\ref{sec:model}, starting with the single-histogram case before extending it to accommodate multiple histograms, similar in concept to the nested DP and the hierarchical DP. Some theoretical properties are also developed. Section~\ref{sec:postsamp} delves into the computational strategies used for inference and Section~\ref{sec:simulations} presents some simulation studies, emphasizing the scalability and efficiency of our methods. To demonstrate the practical value of the approach, we present two applications in Section~\ref{sec:applications}, one using U.S. flight time data and the other U.S. COVID-19 mortality data. We conclude in Section~\ref{sec:disc} by discussing the implications of our work and some thoughts for future research.

\section{Models}
\label{sec:model}

In this section, we describe the two models that are fundamental to this article. First, we introduce a model for estimating the population distribution from a single histogram, which may be regarded as a sufficient statistic for the data. Next, given multiple histograms, we propose a hierarchical model to estimate the population distribution for each histogram. This model allows for the possibility that two or more histograms were generated from a common distribution. In this way, we equivalently address the problem of clustering the underlying distributions which have generated the data binned into different histograms, or of assessing homogeneity between groups.
We conclude the section with additional details relevant to the models.

\subsection{Model for a Single Histogram}
\label{sec:singleHist}

Suppose there exists a vector $\bm{y} = (y_{1}, \ldots, y_{n})$ containing observations such that
$y_j \,|\, G \iid G$ for $j = 1, \ldots, n$, where $G$ is an unknown cumulative
distribution function (CDF) that we wish to estimate. However, in many cases the individual values in
$\bm{y}$ are unobserved, and only bin counts are available.

Define a histogram $\mathcal{H} = (\bm{\tau}, \bm{v})$, where $\bm{\tau}$ is a vector of bin boundaries
that partitions $\mathbb{R}$, with $\tau_{m-1} < \tau_m$ for all $m = 1, \ldots, M$.
For convenience, let $\tau_0 \equiv -\infty$ and $\tau_M \equiv \infty$.
Using $\mathbf{I}(\cdot)$ as the indicator function, define
\[
v_m = \sum_{j=1}^n \mathbf{I}\bigl(\tau_{m-1} < y_j \le \tau_m\bigr),
\qquad \text{ for } m = 1, \ldots, M,
\]
where $v_m$ represents the number of observations contained in the $m$th bin.
Without loss of generality, we construct the bins such that $v_1 = 0$ and $v_M = 0$,
and denote the vector of bin counts as
\[
\bm{v} = (0, v_2, v_3, \ldots, v_{M-1}, 0).
\]
Thus, $\mathcal{H}$ is fully specified by $\bm{\tau}$ and $\bm{v}$.
In this work, given $\bm{y}$,  we treat 
 $\bm{\tau}$ as fixed and known, with
$\tau_1 < \min\{\bm{y}\}$ and $\tau_{M-1} > \max\{\bm{y}\}$.

Some additional notation is required before defining the model.
Let $K$ be a random positive integer representing the number of components in the mixture model we introduce below.
Let $\bm{w}$ denote a probability vector of length $K$, satisfying the constraints
$w_k \ge 0$ and $\sum_{k=1}^K w_k = 1$.
Similarly, let $\bm{p}$ be a probability vector of length $M$, with
$p_m \ge 0$ and $\sum_{m=1}^M p_m = 1$.
Finally, let $\Phi(\,\cdot \mid \mu, \sigma^2)$ denote the CDF
of the Normal distribution with mean $\mu$ and variance $\sigma^2$.

Our proposed model for histograms is:
\begin{align}
    \label{eq:singleHistModel}
    \begin{split}
        \bm{v} \,|\, \bm{p} &\sim \text{Multinomial}\left(n, \bm{p}\right) \\
        p_m &= \tilde{G}(\tau_m)-\tilde{G}(\tau_{m-1}) \text{ for } m=1,\ldots,M\\
        \tilde{G}(\cdot) &= \sum_{k=1}^{K} w_k \Phi\left( \, \cdot \mid \mu_k,\sigma^2_k \right)\\
        K &\sim \text{Discrete distribution on } \mathbb{N}^+\\
        \bm{w} \,|\, K &\sim  \text{Dirichlet}\left(\gamma, \ldots, \gamma \right)\\
       \mu_1,\ldots,\mu_K \,|\, K &\sim \text{Continuous distribution on } \mathbb{R}^K,\\
        & \text{ and such that } \mu_1 < \mu_2 < \ldots < \mu_K \\
        \sigma^2_k \,|\, K &\iid \text{Continuous distribution on } \mathbb{R}^+\!\!\!, \text{ for } k=1,\ldots,K.\\
    \end{split}
\end{align}
We include a few clarifying remarks about the model.
First, $\bm{\tau}$ is fixed and known; in general, no constraints on equal bin widths are required. Second, the mixture distribution $\tilde{G}$ is defined using Gaussian kernels, although this assumption
could be relaxed if desired. We note that $\tilde{G}$ is how the true underlying $G$ is modeled, and that $\tilde{G}$ is deterministic given $K$, 
$\bm{w}$, $\bm{\mu}$, and $\bm{\sigma}^2$ .
Finally, specific priors for $K$, 
$\{\mu_k\} \,|\, K$, and $\{\sigma_k^2 \}\,|\, K$
are not given here; however, when needed, these distributions are defined in the simulation studies and applications.

This model is highly flexible and, depending on the prior placed on $K$, may be viewed as a
mixture of mixtures \citep[MFM;][]{miller2018mixture}. It is worth mentioning that our MFM type approach can be easily replaced by alternative ways to construct the weights, such as stick-breaking priors~\citep{Ishwaran2001}, or any other suitable nonparametric mixture. Our choice here is motivated by the good properties of MFMs and by the ease of implementation. Later in the paper, we show that under suitable
regularity conditions, the posterior of $\tilde{G}$ converges to $G$ (i.e., is consistent in $n$ and $M$) for a broad
class of distributions, $G$.

\subsection{Posterior consistency of the binned density estimator}
\label{sec:consistency}

We now study the consistency of a ``binned'' density estimator obtained from the posterior of model \eqref{eq:singleHistModel}. Let \(M = M_n\) be the number of bins; for \(m = 1,\ldots,M\), let \(\ell_m = \tau_m - \tau_{m - 1}\) the length of the $m$-th bin. The binned density is
\[
    g_b(x) = \sum_{m = 1}^M \frac{p_m}{\ell_m}\mathbf{1}(\tau_{m - 1} < x \le \tau_m).
\]

Assume now that the unobserved data \(Y_1,\ldots,Y_n\) are i.i.d. from a true density \(g^*\). 
Without loss of generality, we assume that \(g^*\) is supported on \([0,1]\). The Hellinger distance between two densities \(g_1\) and \(g_2\) on \([0,1]\) is denoted by
\[
    h^2(g_1, g_2) = \frac{1}{2}\int_0^1 \left(\sqrt{g_1(x)} - \sqrt{g_2(x)}\right)^2 dx.
\]
The next theorem establishes a quantitative consistency result for the posterior of $g_b$ at $g^*$.
To this end, we assume the following conditions hold:
\begin{enumerate}
\item The true density \(g^*\) is \(\beta\)-Hölder continuous on \([0,1]\), for some \(\beta > 0\), and there exist constants \(0 < c_- < c_+ < \infty\) such that \(c_- \le g^*(x) \le c_+\) for every \(x \in [0,1]\).

\item The binning refines with \(n\) and is uniform, that is, \(\tau_m = m / M_n\) for \(m = 0,\ldots,M_n\), and
\[
M_n \asymp \left(\frac{n}{\log n}\right)^{1 / (2s + 1)}.
\]

\item In model \eqref{eq:singleHistModel}, the prior on \(K\) has unbounded support and there exists \(C_K < \infty\) such that
\[
\Pi(K = k) \ge \exp(-C_K k \log k)
\]
for all sufficiently large \(k\). Conditional on \(K\), the ordered means are the order statistics of \(K\) i.i.d. draws from a density \(g_\mu\) such that \(g_\mu\) is bounded away from zero on an open interval containing \([0,1]\). Finally, if \(\sigma_k\) denotes the component standard deviation, the prior for \(\sigma_k\) satisfies, for some constants \(c_\sigma > 0\), \(r_\sigma > 0\), and \(\sigma_0 > 0\),
\[
\Pi_\sigma([t,2t]) \ge c_\sigma t^{r_\sigma}
\]
for every \(0 < t < \sigma_0\).
\end{enumerate}

\begin{theorem}
\label{thm:mfm-binned-hellinger}
Under the assumptions above, there exists \(A < \infty\) such that
\[
\Pi\left(h(g_b, g^*) > A\left(\frac{\log n}{n}\right)^{s / (2s + 1)} \mid \bm v\right) \to 0
\]
in \(P_{g^*}\)-probability.
\end{theorem}
The proof of Theorem \ref{thm:mfm-binned-hellinger} (see the Appendix) is based on an application of the general theory of posterior convergence for non-i.i.d. data from \cite{GvV2007noniid}. This is needed because the sequence of multinomial experiments induced by the refining partitions forms a triangular array.

The theorem is stated for the binned density \(g_b\), rather than for the random mixture density underlying \(\widetilde G\). This distinction is important because the likelihood depends on the density only through the bin probabilities \(p_1,\ldots,p_M\). Therefore, the data cannot distinguish between two densities with the same bin masses but different behavior inside the bins. Without additional assumptions preventing oscillations below the bin width, Hellinger contraction of the random density itself need not follow from binned data.

\subsection{Model for Multiple Histograms}
The ideas underlying our proposed hierarchical model are, in some respects, comparable to the
nested DP described in \cite{rodriguez2008nested} and the hierarchical DP in \cite{teh2006hierarchical}. However, in our model,
the underlying items being clustered are discrete counts rather than individual observations.

Suppose we have $N$ histograms indexed by $i = 1, \ldots, N$.
Let $\bm{y}_i = (y_{i,1}, \ldots, y_{i,n_i})$ denote the data vector underlying the $i$th
histogram, where $n_i$ is the length of that vector. We assume that
\[
y_{i,j} \mid G_i \stackrel{\text{iid}}{\sim} G_i, \qquad \text{ for } j = 1, \ldots, n_i,
\]
where $G_i$ is the CDF of an unknown continuous distribution that we want to estimate.

Let $M_i$ denote the number of bins in the $i$th histogram. Define
$\bm{\tau}_i = (\tau_{i,0}, \ldots, \tau_{i,M_i})$ to be the corresponding bin boundaries, where
$\tau_{i,m-1} < \tau_{i,m}$ for all $m = 1, \ldots, M_i$.
For convenience, let $\tau_{i,0} \equiv -\infty$ and $\tau_{i,M_i} \equiv \infty$.
Additionally, define
\[
v_{i,m} = \sum_{j=1}^{n_i} \mathbf{I}\bigl(\tau_{i,m-1} < y_{i,j} \le \tau_{i,m}\bigr),
\qquad \text{ for }m = 1, \ldots, M_i,
\]
so that $v_{i,m}$ is the number of observations contained in the $m$th bin of the $i$th histogram.
Without loss of generality, we construct the bins such that $v_{i,1} = 0$ and $v_{i,M_i} = 0$,
and write
\[
\bm{v}_i := (0, v_{i,2}, v_{i,3}, \ldots, v_{i,M_i-1}, 0).
\]
We denote the $i$th histogram by $\mathcal{H}_i$, which is fully specified by $\bm{v}_i$ and
$\bm{\tau}_i$. Again, we treat $\bm{\tau}_i \mid \bm{y}_i$ as fixed and known.

For notation, let $K_i$ be a random positive integer representing the number of components in the
$i$th mixture model. Let $\bm{w}_i$ be a probability vector of length $K_i$ with elements
$w_{i,k}$ for $k = 1, \ldots, K_i$, satisfying the constraints $w_{i,k} \ge 0$ and
$\sum_{k=1}^{K_i} w_{i,k} = 1$. Similarly, let $\bm{p}_i$ be a probability vector of length $M_i$
with elements $p_{i,m}$ for $m = 1, \ldots, M_i$, such that $p_{i,m} \ge 0$ and
$\sum_{m=1}^{M_i} p_{i,m} = 1$.

The model, in hierarchical form, is:
\begin{align}
    \begin{split}
        \bm{v}_i \,|\, \bm{p}_i \ind \ &\text{Multinomial}\left(n_i, \bm{p}_i\right) \text{ for } i =1,2,\ldots,N\\
        p_{i,m} = \ & \tilde{G}_i(\tau_{i,m})-\tilde{G}_i(\tau_{i,m-1}) \text{ for } m
        =1,\ldots,M_i\\
        \tilde{G}_i(\cdot) = \ &\sum_{k=1}^{K_i} w_{i,k} \Phi\left( \, \cdot \,|\, \mu_{i,k},\sigma^2_{i,k} \right)\\
       \left(K_i, \bm{w}_i,\bm{\mu}_i,\bm{\sigma}^2_i\right) \,|\, F \iid \ &F\\
       F \,|\, \alpha, \bm{\eta} \sim \ &DP(\alpha, Q_{\bm{\eta}})\\
       \alpha \sim  \ & \text{Continuous Distribution on } \mathbb{R}^+\\
       Q_{\bm{\eta}} = \ & p(K^{\star}) \, p(\bm{w}^{\star} \,|\, K^{\star}) \, p(\bm{\mu}^{\star} \,|\, K^{\star}) \, p(\bm{\sigma}^{2 \star} \,|\, K^{\star}) \\
        K^{\star} \sim \ & \text{Discrete distribution on } \mathbb{N}^+\\
        \bm{w}^{\star} \,|\, K^{\star} \sim  \ & \text{Dirichlet}\left(1/K^{\star},\ldots,1/K^{\star} \right)\\
       \mu_1^{\star},\ldots, \mu_{K^{\star}}^{\star} \,|\, K^{\star} \sim \ & \text{Continuous distribution on } \mathbb{R}^{K^{\star}},\\
                & \text{ and such that } \mu_1^{\star}< \mu_2^{\star} < \ldots < \mu_{K^{\star}}^{\star}\\
        \sigma^{2 \star}_k \,|\, K^{\star} \iid \ & \text{Continuous distribution on } \mathbb{R}^+\!\!\!, \text{ for } k=1,\ldots,K^{\star},\\
    \end{split}
    \label{eq:multipleHists}
\end{align}
where $\bm{\eta}$ is a vector of parameters containing $K^{\star}$, $\bm{w}^{\star}$, $\bm{\mu}^{\star}$, and $\sigma^{2 \star}_k$. The baseline distribution $Q_{\bm{\eta}}$, adopted when defining model \eqref{eq:multipleHists}, has the MFM form and extends the model for a single histogram. In other words, data from each individual histogram arises marginally from a distribution that behaves exactly as in \eqref{eq:singleHistModel}.

This model, through the Dirichlet process, allows two (or more) distinct histograms, $\mathcal{H}_i$ and $\mathcal{H}_j$, to share the same underlying distribution, such that $\tilde{G}_i = \tilde{G}_j$.  We demonstrate the capabilities of this model through
simulation studies and an application.

\subsection{Centering and Scaling Histograms}

As a practical consideration, we conclude that the model is more manageable and that posterior
sampling exhibits better mixing when the data are centered and scaled. The models and
accompanying theory remain unchanged regardless of whether centering and scaling are applied.
The primary difference is that centering and scaling facilitate more standardized priors and
proposal distributions, leading to improved mixing in the MCMC implementation.

When individual observations are available, centering and scaling are straightforward.
However, additional details must be specified when working with histogram data. We do not
claim that the approach described here is optimal in any sense, but we find it intuitive and
easily reversible after model fitting. Our centering and scaling procedure for a given
histogram is as follows:
\begin{enumerate}[noitemsep]
    \item Compute the midpoint of each bin with positive counts.
    \item Convert each count into an equivalent number of pseudo-observations, each taking the
    value of the corresponding bin midpoint.
    \item Compute the mean $\mu^{\dagger}$ and standard deviation $\sigma^{\dagger}$ of these
    pseudo-observations.
    \item Construct new bin boundaries $\tau_m^{\dagger} = (\tau_m - \mu^{\dagger})/\sigma^{\dagger}$ for
    $m = 0, 1, \ldots, M$. The bin counts $\bm{v}$ remain unchanged.
\end{enumerate}
To reverse the centering and scaling, one simply evaluates
$\tilde{G}\bigl((x - \mu^{\dagger})/\sigma^{\dagger}\bigr)$ from the model in~\eqref{eq:singleHistModel}.

When comparing multiple histograms using the model in~\eqref{eq:multipleHists}, a practical approach is to aggregate all pseudo-observations
and compute common centering and scaling parameters, $\mu^{\dagger}$ and $\sigma^{\dagger}$, across all
histograms.

\section{Posterior Sampling}\label{sec:postsamp}

In this section we detail the Markov chain Monte Carlo implementation for the models specified in (\ref{eq:singleHistModel}) and (\ref{eq:multipleHists}).

\subsection{Posterior Sampling for a Single Histogram}
One challenge in implementing this model is that allowing the number of components $K$ to vary changes the dimensionality of the parameter vectors $\bm{w}$, $\bm{\mu}$, and $\bm{\sigma}^2$.  To account for this, we employ the Reversible Jump MCMC algorithm described in \cite{richardson1997bayesian} (R\&G) for finite Gaussian mixture models.  With minor modifications to account for the absence of individual-level observations, we closely follow R\&G's approach.  The resulting MCMC sampling scheme is outlined in Algorithm~\ref{alg:histsim}. Each iteration consists of updating the parameters $\bm{w}$, $\bm{\mu}$, and $\bm{\sigma}^2$, followed by proposals which influence $K$: to split or merge mixture components followed by a birth or a death of a mixture component.

Note that the normal proposal distributions in Algorithm~\ref{alg:histsim} are parameterized using a precision matrix, we follow that same convention through the remainder of the article.

\begin{algorithm}[t]
\caption{MCMC Sampling for a Single Histogram}
\label{alg:histsim}
\begin{algorithmic}[1]
\Require Histogram bin counts $\bm{v}$, bin boundaries $\bm{\tau}$, number of posterior samples $T$, tuning parameters ($\gamma_w$, $\gamma_\mu$, $\gamma_\sigma$), and a pmf on $K$.
\State \textbf{Initialization:} Select initial values for the parameters $K$, $\bm{w}$, $\bm{\mu}$, and $\bm{\sigma}^2$ such that $p(K,\bm{w},\bm{\mu},\bm{\sigma}^2\,|\,\bm{v},\bm{\tau}) > 0$.

\For{$t=1$ to $T$}

\State \textbf{Update $\bm{w}$:} Metropolis-Hastings (M-H) proposal \mbox{$\bm{w}^{*} \sim $ Dirichlet$(\gamma_w \, \bm{w}^{(t-1)})$}.

\State \textbf{Update $\bm{\mu}$:} This is a M-H step using proposal $\bm{\mu}^{*} \sim N(\bm{\mu}^{(t-1)}, \gamma_\mu \,\bm{I}_k)$.

\State \textbf{Update $\bm{\sigma}^2$:} This is a M-H step using proposal \mbox{$1/\bm{\sigma}^{2*} \sim N(1/\bm{\sigma}^{2(t-1)}, \gamma_\sigma \,\bm{I}_k)$}.

\State \parbox[t]{\dimexpr\linewidth-\algorithmicindent}{\textbf{Propose a Split or a Merge:} Proposals are the same as in R\&G, with acceptance probabilities defined in (11) of R\&G with one caveat: $P_{alloc}\equiv 1$, since histogram bins are not assigned to specific components.}

\State \parbox[t]{\dimexpr\linewidth-\algorithmicindent}{\textbf{Propose a Birth or a Death:} Proposals are the same as in R\&G, with acceptance probabilities defined in (12) of R\&G with two caveats: we consider all components to be empty and include the likelihoods (proposed/current) into the ratio.}

\EndFor
\end{algorithmic}
\end{algorithm}

A few notes regarding the MCMC sampling scheme.  We must make a few adjustments to the R\&G's approach since we don't assign individual observations to a mixture component.  The modifications are: First, when updating the component parameters we don't have conjugacy so we use Metropolis and Metropolis-Hastings steps. Next, in the merge-split step, the term $P_{alloc}$ in R\&G's (11) doesn't exist in our formulation (therefore it is set to 1).  Finally, in the birth-death step there is no concept of empty components, thus in R\&G's (12) we include the likelihood ratio, since it no longer cancels, additionally, we set the number of empty clusters ($k_0$) to be the current number of components.

\subsection{For Multiple Histograms}

When considering multiple histograms, we model each individually while also allowing for common $\tilde{G}$ using the Dirichlet process (DP).  The DP induces a partition of the $N$ histograms and we denote the cluster label of the $i$th histogram to be $c_i$, and $\bm{c}$ as the vector of all cluster labels.  To account for the fact that different cluster label vectors may represent the same partition, we restrict $\bm{c}$ to be in canonical form, in which, it uniquely represents a partition using cluster-label notation.  We adopt the following canonical labeling of a partition. The first item is assigned to cluster~1. For each subsequent item, its cluster label is chosen as the smallest positive integer consistent with the existing partition: it is assigned label~1 if it belongs to the same cluster as item~1; otherwise, it may be assigned label~2, and so on, creating a new label only when the item does not belong to any previously formed cluster. In this way, if the partition contains $J$ clusters, each cluster label $c_j \in \{1,2,\ldots,J\}$.

Since our approach to clustering histograms is standard, we can use common MCMC approaches at the partition level of the model.  We reference the Chinese Restaurant Process \citep[CRP; see][]{pitman1996some} and its probability mass function (pmf) to ensure we have valid initial values for $\bm{c}$ and $\alpha$.  For posterior sampling of $\bm{c}$, we use the standard one-at-a-time updates of Neal's Algorithm 8 \citep{neal2000markov}.  The details of our MCMC scheme are outlined in Algorithm~\ref{alg:multhist}. 

\begin{algorithm}[h!]
\caption{MCMC Sampling for Multiple Histograms}
\label{alg:multhist}
\begin{algorithmic}[1]
\Require For $i = 1, \ldots, N$, histogram bin counts $\bm{v}_i$, bin boundaries $\bm{\tau}_i$, number of posterior samples $T$, tuning parameters ($\gamma_w$, $\gamma_\mu$, $\gamma_\sigma$), $m$ (a tuning parameter in \citeauthor{neal2000markov}'s Algorithm 8), common pmf for all $K_i$, prior distributions for $\alpha$, $K^{\star}$, $\bm{w}^{\star} \,|\, K^{\star}$, $\bm{\mu}^{\star} \,|\, K^{\star}$, and  $\bm{\sigma}^{2\star} \,|\, K^{\star}$. \vspace{0.01cm}
\State \textbf{Initialization:} For $i = 1, \ldots, N$, select values for $c_i$, $K_i$, $\bm{w}_i$, $\bm{\mu}_i$, $\bm{\sigma}_i^2$, and $\alpha$ such that $p(K_i,\bm{w}_i,\bm{\mu}_i,\bm{\sigma}_i^2\,|\,\bm{v}_i,\bm{\tau}_i) > 0$, $\bm{c}$ in canonical form such that $c_i=c_n$ if and only if $K_i=K_n$, $\bm{w}_i=\bm{w}_n$, $\bm{\mu}_i=\bm{\mu}_n$, and $\bm{\sigma}_i^2=\bm{\sigma}_n^2$, and $p_{crp}(\alpha, \bm{c}) > 0$, where $p_{crp}$ is the pmf of the CRP. \vspace{0.01cm}
\For{$t=1$ to $T$}
\State \textbf{Update $\alpha^{(t)}$:} We use the augmentation method of \cite{escobar1995bayesian}.
\For{$i=1$ to $N$} 
\State \parbox[t]{0.915\linewidth}{%
Update $c_i^{(t)}$ using \citeauthor{neal2000markov}'s Algorithm 8.
Note: the proposed $m$ clusters have associated parameter values
drawn from $Q_{\bm{\eta}}$, i.e., the distributions for
$K^{\star}$, $\bm{w}^{\star} \,|\, K^{\star}$,
$\bm{\mu}^{\star} \,|\, K^{\star}$, and
$\bm{\sigma}^{2\star} \,|\, K^{\star}$.}
\State Put $\bm{c}^{(t)}$ into canonical form.
\EndFor
\For{$l=1$ to $L^{(t)}$} (where $L^{(t)}$ is the number of clusters defined by $\bm{c}^{(t)}$) 
\Statex \hspace{\algorithmicindent} Note: likelihoods in this loop consist of all histograms with cluster label $l$.
\State \parbox[t]{0.91\linewidth}
{\textbf{Simultaneously simulate} $\bm{w}^{(t)}_{i}$ for all $i \in \{1,\ldots,N\}$ such that $c^{(t)}_i=l$.  This is a Metropolis-Hastings step with proposal $\bm{w}^{*}_{i} \sim $ Dirichlet$(\gamma_w \, \bm{w}^{(t-1)}_{i})$.} \vspace{0.01cm}

\State \parbox[t]{0.91\linewidth}{\textbf{Simultaneously simulate} $\bm{\mu}^{(t)}_{i}$ for all $i \in \{1,\ldots,N\}$ such that $c^{(t)}_i=l$.  This is a Metropolis step using proposal $\bm{\mu}^{*}_{i} \sim N(\bm{\mu}^{(t-1)}_{i}, \gamma_\mu \bm{I}_k)$.} \vspace{0.01cm}

\State \parbox[t]{0.91\linewidth}{\textbf{Simultaneously simulate} $\bm{\sigma}^2$ for all $i \in \{1,\ldots,N\}$ such that $c^{(t)}_i=l$. This is a Metropolis step using proposal $1/\bm{\sigma}^{2*}_i \sim N(1/\bm{\sigma}^{2(t-1)}_i, \gamma_\sigma \bm{I}_k)$.} \vspace{0.01cm}

\State \parbox[t]{0.91\linewidth}{\textbf{Propose a Split or a Merge:} For $\tilde{G}_l$, propose a split or merge for one of its mixture components as described in R\&G, with acceptance probabilities given in (11) of R\&G, with one caveat: $P_{alloc}\equiv 1$, since histogram bins are not assigned to specific components. If a split or a merge was accepted, update all $K_i^{(t)}$, $\bm{w}_i^{(t)}$, $\bm{\mu}_i^{(t)}$, $\bm{\sigma}_i^{2(t)}$ for all $i \in \{1,\ldots,N\}$ such that $c^{(t)}_i=l$.} \vspace{0.01cm}

\State \parbox[t]{0.91\linewidth}{\textbf{Propose a Birth or a Death:} For the histogram associated with cluster label $l$, propose the birth or a death of a mixture component as described in R\&G, with acceptance probabilities defined in (12) of R\&G with two caveats: we consider all components to be empty and include the likelihood (proposed/current) into the ratio. If a birth or a death was accepted, update all $K_i^{(t)}$, $\bm{w}_i^{(t)}$, $\bm{\mu}_i^{(t)}$, $\bm{\sigma}_i^{2(t)}$ for all $i \in \{1,\ldots,N\}$ such that $c^{(t)}_i=l$.} \vspace{0.01cm}

\EndFor
\EndFor
\end{algorithmic}
\end{algorithm}

\section{Simulations}\label{sec:simulations}

In this section we conduct simulation studies for the proposed models.  Before considering each model we outline some common details for all of the simulations.

To generate a histogram from sampled data we do the following:  First, sample the appropriate distribution. Next, two small amounts of random noise (both $iid$ from a U$(0,0.1)$) are subtracted from and added to the minimum and the maximum observations, respectively.  The resultant quantities are used to define $\tau_{1}$ and $\tau_{M-1}$ respectively.  The remaining bin boundaries are then equally spaced between $\tau_{1}$ and $\tau_{M-1}$.

For posterior sampling in each scenario we obtain 200,000 draws from the posterior distribution.  For burn-in we discard the first half of the chain, and for thinning we keep every 10th draw.  The resulting number of posterior samples is 10,000.

\subsection{For a Single Histogram} \label{sec:singleHistSims}

To demonstrate the properties of our method we considered the well known galaxy dataset \citep{postman1986probes,roeder1990density}.  After centering and scaling the data, we construct a plausible mixture of four normal distributions which could have generated the data and which we use as the ``true'' distribution for the simulations in this section.  This contrived CDF is:
\begin{align} \label{eq:true}
\begin{split}
F(x) =  
&\frac{7}{82}\,\Phi\left(\frac{x+2.436}{0.1}\right)+
\frac{13}{82}\,\Phi\left(\frac{x+0.25}{0.1}\right)+\\
&\frac{59}{82}\,\Phi\left(\frac{x-0.35}{0.5}\right)+
\frac{3}{82}\,\Phi\left(\frac{x-2.677}{0.25}\right),
\end{split}
\end{align} 
where $\Phi(\cdot)$ is the CDF of the standard normal distribution.  

The following priors were used for the simulations in Section \ref{sec:singleHistSims}.
\begin{align*}
    K &\sim  \text{Geometric}(0.15) \times \text{I}\,(1 \leq K \leq 30)\\
    \bm{w} \mid K &\sim \text{Dirichlet}(1,\ldots,1)\\
    \bm{\mu} \,|\, K &\sim \text{Normal}\left(\bm{0},0.01 \times \bm{I}_{K}\right) \times \text{I}(\mu_1 < \mu_2 < \cdots < \mu_{K}) \\
   1/ \sigma^{2}_k \mid K &\iid \  \text{Gamma}(1,0.01) \text{ for } k=1,\ldots,K.
\end{align*}
We specified the prior on $K$ to roughly approximate the average number of clusters for the CRP with mass parameter set to 1.  The other priors were chosen to be sufficiently diffuse for these scenarios and are dependent on the assumption that the data are roughly centered and scaled.

The computations in this article were accomplished on a Linux server with 2 Intel Xeon 6787P processors with a total of 503GB of RAM.

In this first simulation we assumed several sample sizes and number of possible bins and determine how close we estimate the truth using the model in \eqref{eq:singleHistModel}.  To determine ``closeness'' between two measures $\nu_1$ and $\nu_2$ 
we use the Wasserstein-$1$ distance which can be defined in a one dimensional setting as:
\begin{equation} \label{eq:wass}
W_1(\nu_1,\nu_2) = \int_{-\infty}^{\infty} \big|\, F_1(x)-F_2(x) \,\big|\ dx
\end{equation}
\citep{wassDist}, where $F_1$ and $F_2$ are the respective CDFs.  Additionally, the chosen sample sizes were 100, 200, 400, 800, 1600, 3200, and 6400.  The number of bins varied from 5 to 50 by increments of 5.  For each of these samples sizes and number of bins, 30 datasets were randomly generated.  For each dataset a model was fit, and the average posterior CDF was calculated from the posterior samples.  Specifically, each posterior draw produced a mixture distribution with an associated CDF, and all posterior CDFs were averaged to estimate the mean posterior CDF.  Using this mean posterior CDF, the Wasserstein distance was calculated with respect to the true data generating distribution given in \eqref{eq:true}.  The average distance for the 30 replications of each sample size and number of bins represents a single point in the following plots. 

Figure \ref{fig:Converge2True2} shows the results of the simulation.  For a fixed sample size, the average distance from the true CDF initially decreases from 5 bins to 10 bins, and from 10 bins to 15 bins.  However, after increasing the number of bins to 25 or 30, the average distance from the truth doesn't show meaningful improvement. 
The plot also shows that as the sample size increases, the distance from the truth decreases (as one might expect).  Thus more data appears to produce a better mixture model density estimate.

\begin{figure}[ht!]
    \centering
    \includegraphics[width=.5\paperwidth]{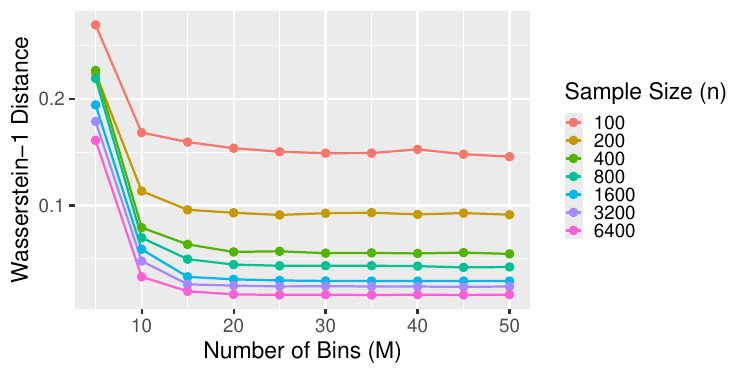}
    \caption{The results of a simulation study which varied the sample size and number of bins.  The results are compared with the Wasserstein's distance from the true data generating distribution.  It is evident that a higher number of samples improves performance and, to a point, more bins also produces better performance.}
    \label{fig:Converge2True2}
\end{figure}

The same simulation results are shown in Figure \ref{fig:Converge2True1} with the lines and x-axis changing roles.  From this view, we see that the performance of 5 bins is quite bad, but as the number of bins increases, the performance stabilizes.  It is difficult to distinguish any meaningful difference in the performance between bin sizes of 25-50. 

\begin{figure}[ht!]
    \centering
    \includegraphics[width=.5\paperwidth]{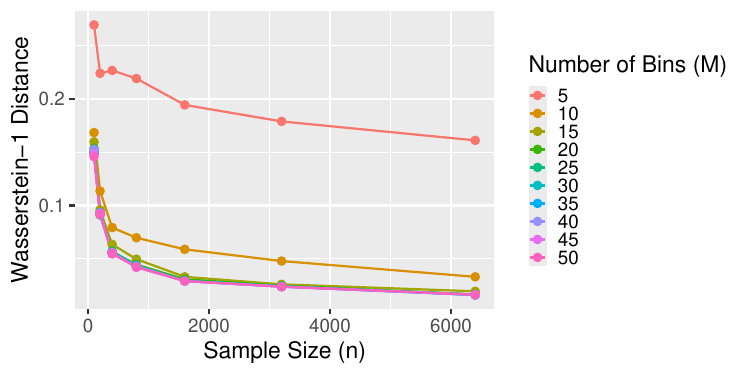}
    \caption{Similar to Figure \ref{fig:Converge2True2}, this plot shows that using more bins allows for improved performance with larger samples.  With a large number of samples, the performance of 25-50 bins is difficult to distinguish.}
    \label{fig:Converge2True1}
\end{figure}

The results of this simulation study indicate that the model appears to be consistent (as shown in Section \ref{sec:consistency}) as the number of samples and bins increase.  Additionally, from a practical perspective, the number of bins has a significant effect on the model's accuracy for a given sample size. 

\subsubsection{Comparison with the DPM Model}
\label{sec:DPMComp}

In this section we compare our model with the traditional Dirichlet process mixture (DPM) model to the original data.  This is an interesting comparison in that our model and the DPM model are very similar.  The primary difference is that the DPM model uses all the data, whereas our proposed model only uses the binned data.  Using our model, we demonstrate how the computational time scales much better (in $n$) than the DPM, and that the loss of accuracy can be minimal.

To date, fitting a DPM to a very large dataset is not computationally feasible.  However, our method scales extremely well with sample size.  Using the previous simulation scenario, we compared our method to a DPM model analyzing the same (but unbinned) data using the BNPmix R package \citep{BNPmix}. 
Again, using 30 replicates for each sample size we plot the averaged results.
The timing numbers are shown in Figure \ref{fig:CompWithDPM1}.  One can observe with each horizontal dashed line, that as the number of samples increases, the 
computational time for the DPM model steadily moves up.  The computational time of our method, shown in solid lines, appears very similar for each of the four sample sizes.  However, the computational time of our method increases, but only as a function of the number of bins, as shown by the increasing slope of the solid lines.  
It is important to note that, in this case, the computational speed of our method does not appear to be meaningfully affected by the sample size.

\begin{figure}[ht]
    \centering
    \includegraphics[width=.5\paperwidth]{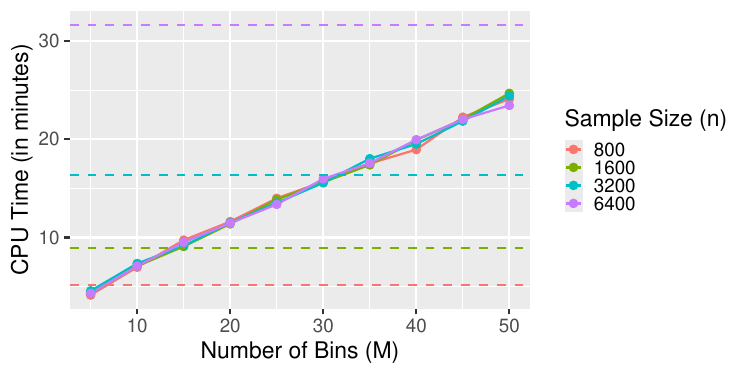}
    \caption{Timing comparisons with the DPM at four different sample sizes.  As the number of samples goes up the DPM's computation increases (dashed lines), however, our method is nearly invariant to the sample size, but its computation time increases as the number of bins increases (solid lines).}
    \label{fig:CompWithDPM1}
\end{figure}

Although the timing numbers look promising, the model accuracy is also of prime consideration.  Figure \ref{fig:CompWithDPMAccuracy1} shows the accuracy of our model (solid lines) compared with the DPM model.  It is clear, as long as there are a sufficient number of samples and bins, the accuracy of our model is not an issue.  And under certain conditions, it appears to have similar accuracy as the DPM model (dashed lines). 

\begin{figure}[ht]
    \centering
    \includegraphics[width=.5\paperwidth]{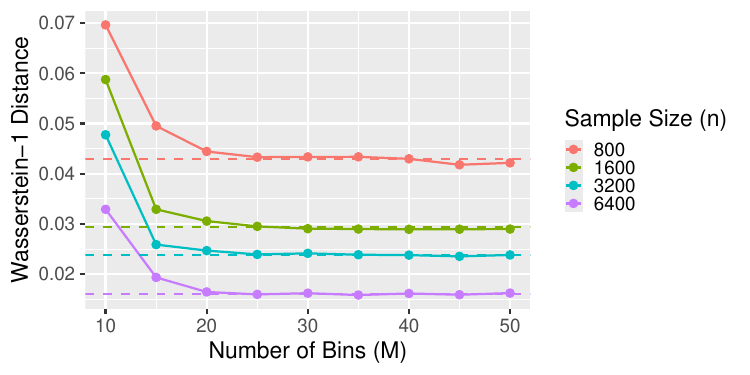}
    \caption{Comparisons of accuracy with the DPM model (dashed lines) at four different sample sizes.  Clearly with sufficient data and number of bins our method (solid lines) achieves the accuracy of the DPM model.}
    \label{fig:CompWithDPMAccuracy1}
\end{figure}

Figure \ref{fig:CompWithDPMAccuracy1} clearly shows that using a sufficient number of bins is important to model behavior, but in this example, the gains in accuracy tail off after using 30 or more bins when the sample size is sufficiently large.

\subsubsection{Misspecified Comparison with the DPM Model}

In the previous simulations the ``true'' model was assumed to be a mixture of Gaussian kernels.  In these simulations we explore the effect of modeling a mixture distribution using a mixture of non-Gaussian kernels.  
The ``true'' model is assumed to have a density of:
\[f(x) = .4 \times Gamma(x+5,6,3) + .2 \times t(x,5) + .4 \times t(x-4,3). \]
Using this model as the data generating distribution, the model in \eqref{eq:singleHistModel} is now misspecified in the kernel.

We consider the same comparisons with the DPM as in Section \ref{sec:DPMComp}. Figure \ref{fig:CompWithDPM} shows that the sample size now appears to play a role in the computation time.  As the bins increase we see a separation of the solid lines.  A major factor in this phenomenon is that, as the sample size $n$ increases, more mixture components $K$ are needed to better describe the asymmetry and the thicker tails of the ``true'' model.  So although the sample size influences the computation time for our proposed model, the primary factor in computing time is still the number of bins.  We also notice that for the sample sizes of 3200 and 6400, the computation times for the DPM have noticeably increased (compared to Figure \ref{fig:CompWithDPM1}).
\begin{figure}[ht]
    \centering
    \includegraphics[width=.5\paperwidth]{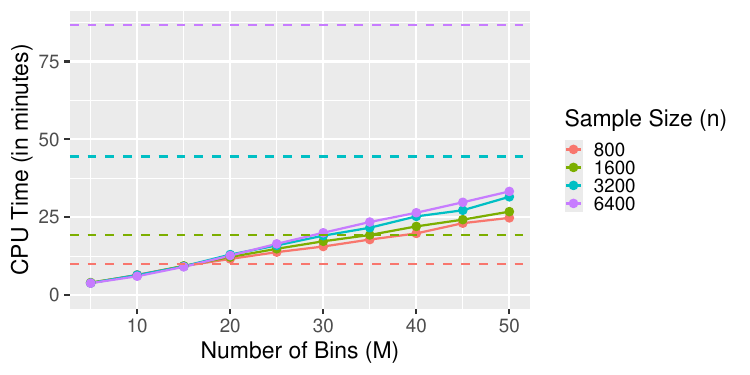}
    \caption{Timing comparisons with the DPM for the misspecified case using four different sample sizes.  Clearly as the number of samples goes up the computation time (in minutes) for the DPM increases, however, our method is somewhat invariant to the sample size, but its computation slows as the number of bins increase.}
    \label{fig:CompWithDPM}
\end{figure}

In regards to accuracy with a misspecified kernel in the model, our model's accuracy is slower to match the accuracy of the DPM. This is shown in Figure \ref{fig:CompWithDPMAccuracy}.  So while the solid lines are approaching the dashed lines (of similar colors), they do not match as closely as with the case of the correctly specified kernels.

\begin{figure}[ht]
    \centering
    \includegraphics[width=.5\paperwidth]{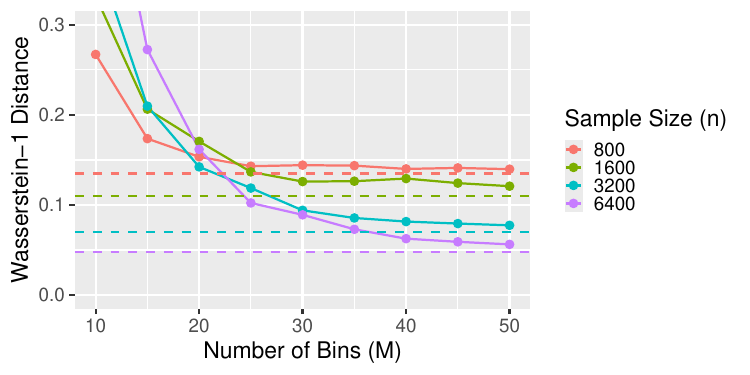}
    \caption{Comparisons of accuracy with the DPM model for the misspecified case using four different sample sizes.  Clearly with sufficient data and number of bins our method begins to match the accuracy of the DPM model.}
    \label{fig:CompWithDPMAccuracy}
\end{figure}

\subsection{Clustering Multiple Histograms} \label{sec:simsMultHists}

We evaluate the performance of the hierarchical model in \eqref{eq:multipleHists} for clustering multiple histograms. The primary objectives of this simulation are to assess the model’s ability to recover the true clustering structure while simultaneously estimating the underlying population distributions. We expect the model performance to primarily be a function of three things: the sample size, the number of bins, the number of histograms being clustered.  Generally, computation time will slow as these quantities increase, but the accuracy should improve.

The following metrics provide a meaningful way to assess model performance. 
To determine how well the model estimates the density for a given histogram we again consider the Wasserstein-1 distance, as defined in \eqref{eq:wass}.  To measure how well the model clusters similar histograms, we use the adjusted Rand index (ARI) (\citealt{hubert1985comparing}). The ARI determines the similarity of two given partitions, with values ranging between -1 and 1, with a value of 1 denoting two identical partitions.  An ARI value of 0 suggests partition agreement  equivalent to random assignment.

In our simulation scenario we consider 100 histograms.  The goal is to model their underlying generating densities and determine if any histograms are generated from the same population, by clustering them together.  
Each true population density is the mixture of two normal distributions as follows:
\begin{align} \label{eq:trueMultHists}
\begin{split}
F_1(x) &= 0.4\,\Phi\left(\frac{x+1}{0.8}\right)+
0.6\,\Phi\left(\frac{x-1}{0.8}\right),\\
F_2(x) &= 0.6\,\Phi\left(\frac{x+1}{0.8}\right)+
0.4\,\Phi\left(\frac{x-1}{0.8}\right),\\
F_3(x) &= 0.5\,\Phi\left(x-1\right)+
0.5\,\Phi\left(x-3\right),\\
F_4(x) &= 0.5\,\Phi\left(x+1\right)+
0.5\,\Phi\left(x+3\right).
\end{split}
\end{align} 
The simulation is designed by generating 10 histograms from Population 1 (i.e., $F_1(x)$ in \eqref{eq:trueMultHists}), 40 histograms from Populations 2 and 4 and the final 10 histograms from Population 3.  The densities for these populations are shown in Figure \ref{fig:trueDensities}.

Here we specify the choices of the priors and hyperparameters for the model in \eqref{eq:multipleHists} for this simulation study.  They are:
\begin{align*}
    \alpha &\sim \text{Gamma}(2,4)\\
    K^{\star} &\sim  \text{Geometric}(0.5) \times \text{I}\,(1 \leq K^{\star} \leq 45)\\
    \bm{w}^{\star} \mid K^{\star} &\sim \text{Dirichlet}(1,\ldots,1)\\
    \bm{\mu}^{\star} \,|\, K^{\star} &\sim \text{Normal}\left(\bm{0},0.01 \, \bm{I}_{K^{\star}}\right) \times \text{I}(\mu_1^{\star} < \mu_2^{\star} < \cdots < \mu_{K^{\star}}^{\star}) \\
   1/ \sigma^{2 \star}_k \mid K^{\star} &\iid \  \text{Gamma}(1,0.01) \text{ for } k=1,\ldots,K^{\star}.
\end{align*}

\begin{figure}[ht]
    \centering
    \includegraphics[width=.615\paperwidth]{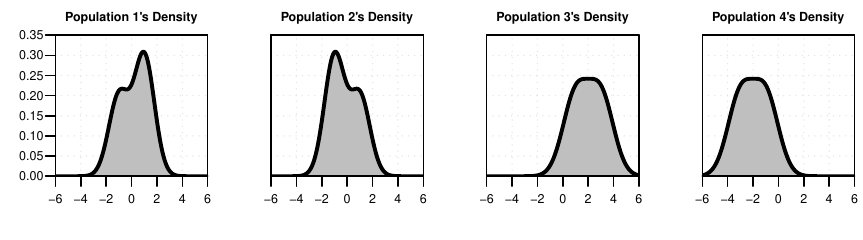}
    \caption{The four distributions which generated the data for this simulation study.}
    \label{fig:trueDensities}
\end{figure}

We specified 7 sample sizes and number of bins configurations, each defining a separate row in Table \ref{tab:clusteringSim}.  On average, we expect the model accuracy to improve as both the number of bins and sample size increases.  For each sample size and number of bin configurations we ran the model 30 times and took the mean of the results.

The four center columns in Table \ref{tab:clusteringSim} show the mean Wasserstein-1 distance for each histogram's posterior mean CDF estimate to its true generating CDF.  As expected, the distance to the true distributions generally decrease with more samples and more bins.  Additionally, the populations that have 40 histograms associated (versus 10)  appear to have smaller distances to the truth, which we attribute to more sharing of information within a cluster.  In the second to last column we consider the ARI between the random partition induced by sampling from a DP in \eqref{eq:multipleHists} and the true clustering (the 4 population clusters).
The mean posterior ARI is shown, and it also typically improves with more samples and more bins.  
The mean time to sample from the posterior is given in the last column of Table \ref{tab:clusteringSim}.  

\begin{table}
    \centering
    \caption{The results from simulating each cell 30 times and taking the mean values. As can be seen, the distance from the true distribution decreases for each population as the sample size and the number of bins increases.  Also, the clustering of like populations improves with an increase of sample size and number of bins as demonstrated with the mean adjusted Rand index (ARI).}
    \label{tab:clusteringSim}
    \begin{tabular}{cc|cccc|cc}
    & & \multicolumn{4}{c|}{Wasserstein-1 Distance} & \multicolumn{2}{c}{Mean} \\
    Sample & Number & \multicolumn{4}{c|}{Population:} & Posterior & Time\\
    Size & of Bins & 1 & 2 & 3 & 4 & ARI & (hours) \\ \hline
    100 &   10 & 0.297 & 0.078 & 0.083 & 0.078 & 0.847 & 8.4 \\ 
    100 &   15 & 0.247 & 0.091 & 0.080 & 0.078 & 0.877 & 9.3 \\ 
    1,000 &   15 & 0.012 & 0.007 & 0.015 & 0.008 & 1.000 & 11.9 \\ 
    1,000 &   20 & 0.012 & 0.006 & 0.014 & 0.008 & 1.000 & 15.8 \\ 
    10,000 &   20 & 0.010 & 0.012 & 0.010 & 0.011 & 0.999 & 44.1 \\ 
    10,000 &   25 & 0.014 & 0.013 & 0.013 & 0.015 & 1.000 & 61.2 \\ 
    10,0000 &   25 & 0.009 & 0.008 & 0.009 & 0.009 & 0.982 & 70.1 \\ 
      \hline
    \end{tabular}
\end{table}

The results from this simulation study demonstrate that the proposed framework effectively clusters histogram-valued data while recovering the underlying distributions. The Dirichlet process prior facilitates information sharing across histograms, and in principle, improves performance over considering each histogram in isolation. As demonstrated, the method is robust to heterogeneous binning boundaries, making it suitable for applications involving aggregated data.

\section{Applications}\label{sec:applications}

In this section we demonstrate our method in two scenarios.  In the first, we consider modeling the density of only one histogram, but with a large amount of data.  The second application considers many histograms in which we estimate the population density and, in the hierarchical fashion of the model in (\ref{eq:multipleHists}), allow some sharing of estimated population densities among the histograms.  These two applications demonstrate the usefulness of the models in (\ref{eq:singleHistModel}) and (\ref{eq:multipleHists}) proposed in this paper.

\subsection{U.S. Airline Flight Times}
The U.S. Bureau of Transportation Statistics (BTS) publishes monthly data which tracks the flight time for all domestic flights in the United States.  We consider a four month window from January to April 2025.  As one might imagine, there are many domestic U.S. flights in one month.  During this four month time-frame there were roughly 2.23 million flights, of which roughly 2.19 million have a recorded flight time.  These flight times are recorded in whole minutes.  The flights ranged in length from 8 minutes to 715 minutes. A barchart of the flights times less than 350 minutes is shown in Figure \ref{fig:histVSpost} with thin black vertical lines.  One can see from this plot, these flight times are multi-modal and skewed to the right.

To model the density of these flight times using the traditional DPM model is very time consuming.  However, using our method, we can fit an accurate mixture model to these data in a reasonably fast time.  This application demonstrates our method's ability to model a large amount of data using binning.  We consider several versions of our model based on the number of bins to show the time and accuracy trade-off for more or fewer bins.

We employ the model in \eqref{eq:singleHistModel} and the choices for distributions on $K$, $\bm{w} \,|\, K$, $\bm{\mu} \,|\, K$, and $\bm{\sigma}^2 \,|\, K$ are:
\begin{align}
    \label{eq:FlightTimes}
    \begin{split}
        K &\sim \text{Geometric}(0.5) \times \text{I}\,(1 \leq K \leq 30)\\
        \bm{w} \,|\, K &\sim  \text{Dirichlet}\left(1/K,\ldots,1/K \right)\\
       \bm{\mu} \,|\, K &\sim \text{Normal}\left(\bm{0},0.01 \, \bm{I}_K\right) \times \text{I}(\mu_1 < \mu_2 < \cdots < \mu_K) \\
        1/\sigma^2_k \,|\, K &\iid \text{Gamma}(1,0.01), \text{ for } k=1,\ldots,K.\\
    \end{split}
\end{align}

There are 680 unique flight time values, each recorded in whole minutes.  Because we have the actual observations, we can experiment somewhat with accuracy and computational trade-offs.  We consider five different size configurations for the bins $\bm{\tau}$.  With no loss of data fidelity we start with 723 bins, then combining different values of observations 363 bins, 183 bins, 93 bins and finally 48 bins.  For each scenario we specify $\bm{\tau}$ as follows:
\begin{align*}
    \text{723 bins: }\ \bm{\tau} &= (-\infty,-0.5,0.5,1.5,\ldots,718.5,719.5,\infty),\\
    \text{363 bins: }\ \bm{\tau} &= (-\infty,-0.5,1.5,3.5,\ldots,717.5,719.5,\infty),\\
    \text{183 bins: }\ \bm{\tau} &= (-\infty,-0.5,3.5,7.5,\ldots,715.5,719.5,\infty),\\
    \text{93 bins: }\ \bm{\tau} &= (-\infty,-0.5,7.5,15.5,\ldots,711.5,719.5,\infty),\\
    \text{48 bins: }\ \bm{\tau} &= (-\infty,-0.5,15.5,31.5,\ldots,703.5,719.5,\infty).\\
\end{align*}
\begin{figure}[ht]
    \centering
    \includegraphics[width=.45\paperwidth,trim=.15cm .5cm 1cm 2cm,
    clip]{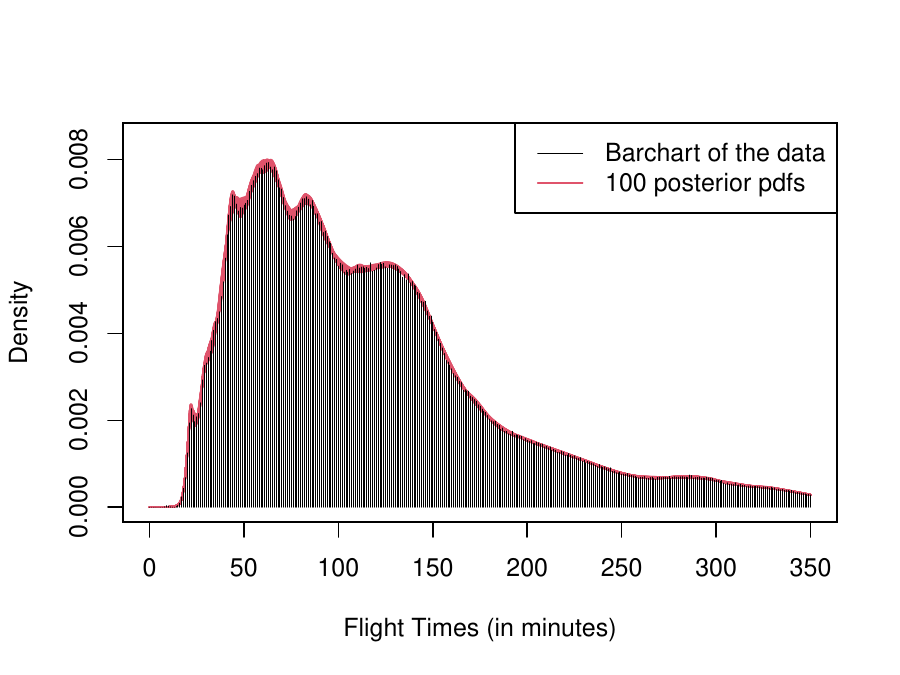}
    \caption{A comparison of a barchart of the flight times less than 350, in black, compared with 100 posterior draws of the fitted mixture model pdf (utilizing 723 bins), in red.}
    \label{fig:histVSpost}
\end{figure}
The bin counts are computed for each $\bm{\tau}$.  Following which, the $\bm{\tau}$ vectors are then centered and scaled by subtracting 120 and dividing that result by 72.

We ran 10 chains for each model. Similar to the simulations, for each chain we obtained 200,000 posterior draws, discarded half for burn-in, and then thinned by keeping every 100th sample.  Combining the 10 chains resulted in 10,000 posterior samples.

\begin{figure}[ht]
    \centering
    \begin{subfigure}{0.49\linewidth}
        \centering
        \includegraphics[width=\linewidth,trim=.1cm .5cm .1cm 1.5cm,
    clip]{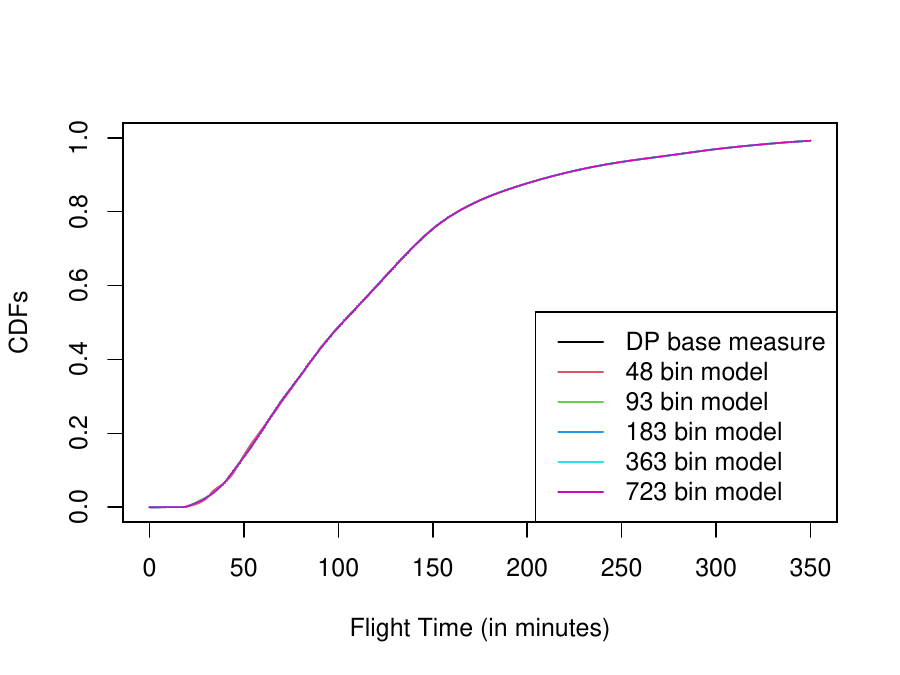}
        \caption{CDF estimates of flight time data.}
        \label{fig:ecdfVSpost1}
    \end{subfigure}
    \hfill
    \begin{subfigure}{0.49\linewidth}
        \centering
        \includegraphics[width=\linewidth,trim=.1cm .5cm .1cm 1.5cm,
    clip]{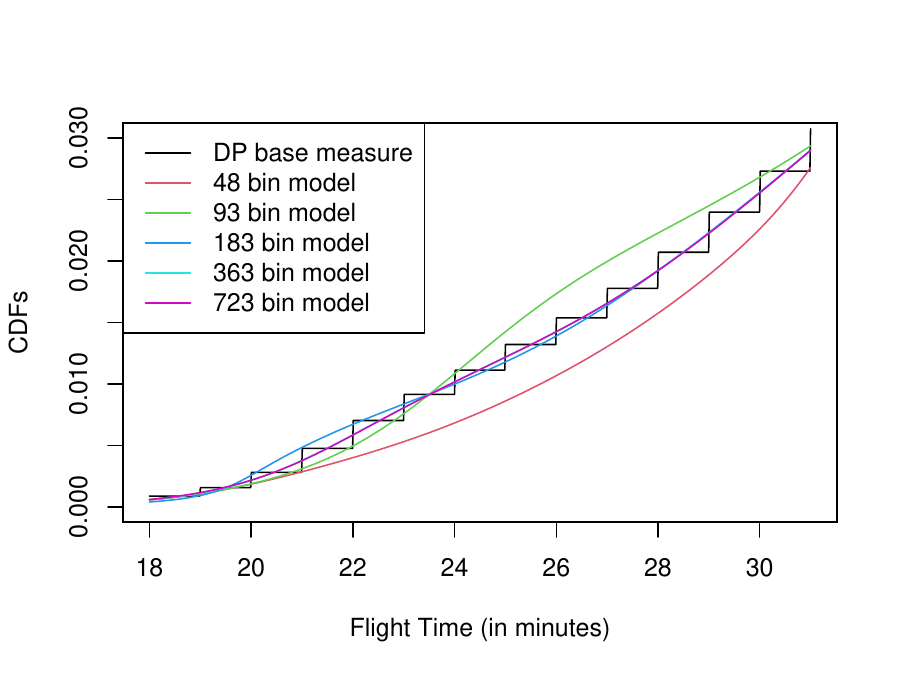}
        \caption{Zoomed-in view of left plot.}
        \label{fig:ecdfVSpost2}
    \end{subfigure}
    \caption{Comparison of the DP posterior base measure to the mean posterior CDF using the model in (\ref{eq:singleHistModel}) with 5 different bin sizes.  In the left plot it appears that all the models are in fairly close agreement.  The left plot is a zoomed-in view and shows some differences between the 5 models (in color), where the DP is in black.  As the number of bins increases, the closer the models get to a smoothed version of the DP.}
    \label{fig:ecdfVSpost}
\end{figure}

Some results for the model with 723 bins are displayed in Figure \ref{fig:histVSpost}.  The figure shows a barchart of the raw data in black (with flight times less than 350 minutes), plotted with 100 draws of the posterior mixture density in red lines. 
At a glance the model fits the data quite well, and captures the multimodal and skewed shape.  As expected, some of the individual spikes are smoothed over. 

In Figure \ref{fig:ecdfVSpost} we compare the mean posterior CDFs for the five different binning schemes with the posterior base measure of a DP fit directly to the data.  In other words, $x_i\,|\,G \iid G$ for $i\in\{1,\ldots, 2188776\}$, and $G\sim DP(\alpha, G_0)$.  The posterior of $G\,|\,\bm{x}$ is also a DP with base measure:
\begin{equation*} \label{PosteriorDP}
    \frac{\alpha G_0(t) + n \hat{F}(t)}{\alpha+n},
\end{equation*}
where we selected $\alpha=1$, $G_0(t)$ is the CDF of  $\text{Gamma}(1,1/55)$, $n=2,188,776$, $\hat{F}(t)$ is the empirical CDF, and $t$ represents time in minutes.  
This comparison demonstrates that our model, across several possible binning schemes, can capture the behavior of the estimated distribution in a manner comparable to a largely empirical approach, while importantly providing a smooth estimate.

It is difficult to detect any difference between the CDFs in Figure \ref{fig:ecdfVSpost1}, with the exception of some very slight bumps before time 50 (where a small mode exists).  A zoomed-in version of the CDFs is displayed in Figure \ref{fig:ecdfVSpost2}.  At this resolution, we can detect that the models with fewer bins have more difficulty capturing the more subtle behavior of the data.  As the number of bins is increased the model accuracy also appears to improve.

\begin{table}
    \centering
    \caption{Comparison of five different bin size configurations for the model in (\ref{eq:singleHistModel}) and the DPM model for the U.S. flight time data.  The \textit{CPU Time} is the average minutes in CPU time used to compute 200,000 posterior samples (the lower bound for the DPM estimate is an extrapolation from only 1,000 samples).  $E[K]$ and SD$(K)$ are the average number of posterior clusters and the standard deviation, respectively.}
    \label{tab:fltResults}
    \begin{tabular}{c|ccc}
      Number of Bins& CPU Time (minutes)  &  $E[K]$ & SD$(K)$\\ \hline
       48 & 47.9 & 17.7 & 2.1 \\
       93 & 102.7 & 21.6 & 2.6 \\
       183 & 200.8 & 23.3 & 2.2\\
       363 & 403.0 & 24.6 & 3.2 \\
       723 & 719.2 & 22.6 & 2.4 \\ 
       $\infty$ (DPM) & $>$3,000 & --- & --- \\\hline
    \end{tabular}
\end{table}

The remainder of the results are contained in Table \ref{tab:fltResults}.  The average CPU time for each chain is reported along with the posterior mean of $K$ and its standard deviation.  We also aimed at comparing with DPM model as we did for the simulated data. However, timing for the DPM model was difficult to estimate.  When running 1,000 samples we measured a little over 1.5 seconds per MCMC iteration.  However, the chains had not converged and the number of mixture components $K$ was lower than we'd expect---this has a dramatic effect on the timing.  When trying to run more than 2,000 iterations, the software would crash due to memory problems.  Therefore, the estimated CPU time for the DPM model is an underestimate.  This is a clear example where the standard BNP approach has problems due to the large amount of data.  However, our model does well, it includes uncertainty, and was computed in the range of 1 to 12 hours (depending on the number of bins).  For timing considerations, it is  important to note that the BNPmix package is using compiled code and our code is not.  

\subsection{U.S. COVID-19 Mortality Counts}

The U.S. Center for Disease Control (CDC) published the COVID-19 mortality counts for each state.  According to the CDC website, the data contain ``deaths involving COVID-19, pneumonia, and influenza reported to [National Center for Health Statistics] by sex, age group, and jurisdiction of occurrence.''
These data are totals from the beginning of the pandemic until September 27, 2023.  They reported counts by state, age group, and sex in addition to other information.  In total, there were more than 1.1 million deaths reported.  We consider the following histograms.  Mortality counts for the following age groups:
0-17 years,
18-29 years,
30-39 years,
40-49 years,
50-64 years,
65-74 years,
75-84 years,
85 years and over.
Data were collected for all 50 U.S. states in addition to the District of Columbia, Puerto Rico, and New York City.  In the data, there were 106 tables of mortality counts (53 geographical areas multiplied by 2, for females and males).  For ease of terminology we will use the term \textit{state} for the geographical regions although District of Columbia, Puerto Rico, and New York City are not U.S. states.

It is of note that the CDC did not report mortality counts for a bin when numbers were below 10.  However, we were able to solve for some missing values by deduction from the totals.  For those that we could not ascertain, we combined bins and bin counts.  This approach resulted in 94 state/sex groups having all 8 bins available, 11 groups having 7 bins, and 1 group with only 6 bins.  Additionally, to better control behavior of the right tail of the mixture model we created a somewhat arbitrary upper bound of age 100 for the 85 years and over bin.  The resulting mixture models are somewhat sensitive to this choice, however, 100 seemed to be reasonable given that less than 0.03\% of the U.S. population is older than 100 \citep{pew_centennials_2024}. Bins for ages less than 0 and greater than 100 are forced to have counts of zero.  For this application we used common centering and scaling constants for all histograms to ensure we didn't degrade the model's ability to detect differences in distributional shifts.  Code for data preparation is available upon request. 

Although one could model the data only by state, we felt that sex was an important variable that should be used to better explain the data.  Overall, the reported CDC mortality numbers were roughly 52K females and 63K males during that time-frame.  Additionally, using the bin midpoints, the average age at death for females and males was 76.6 and 72.9 years old, respectively.  With these differences in mind, it seemed appropriate to differentiate the data by sex. 

In this application $N=106$, $M_i=8$ for all state/sex groups except 12 (as described above). 
We apply the model given in \eqref{eq:multipleHists}.  Specific choices were made for several of the priors and hyperparameters: 
\begin{align*}
    \alpha &\sim \text{Gamma}(2,4)\\
    K^{\star} &\sim  \text{Geometric}(0.9) \times \text{I}\,(1 \leq K^{\star} \leq 8)\\
    \bm{w}^{\star} \mid K^{\star} &\sim \text{Dirichlet}(1,\ldots,1)\\
    \bm{\mu}^{\star} \,|\, K^{\star} &\sim \text{Normal}\left(\bm{0},0.01 \, \bm{I}_{K^{\star}}\right) \times \text{I}(\mu_1^{\star} < \mu_2^{\star} < \cdots < \mu_{K^{\star}}^{\star}) \\
   1/ \sigma^{2 \star}_k \mid K^{\star} &\iid \  \text{Gamma}(1,1) \text{ for } k=1,\ldots,K^{\star}
\end{align*}
Since the max number of bins, with a non-zero count, is 8, we bound $K^{\star}$ in $\{1,2,\ldots,8\}$ to ensure there is at most 8 mixture components for every histogram, thus alleviating problems related to overfitting. 

To obtain 200,000 posterior samples for one chain it took an average of roughly 21 hours. Half the samples were discarded for burn-in, and for thinning, every 100th draw was kept.  We ran 10 chains and obtained a total of 10,000 posterior samples.   

On average there were about 6 clusters in the posterior partition, with a min and a max of 4 and 8, respectively.  We found a point estimate of the partition by using the default settings of \textit{salso} \citep{dahl2022search}.  
The resulting partition is shown in Figure \ref{fig:sidebysideUS} (one map is for female data and the other for male data).  
It was not practical to display the District of Columbia, New York City, and Puerto Rico on the maps, however, they were assigned to clusters: 2, 4, and 4 for females and 1, 2, and 2 for males, respectively (see the legend in Figure \ref{fig:sidebysideUS}).  It is clear from both maps that there are spatial relationships.  Additionally, there are clear sex effects on mortality---no states have both females and males in the same cluster. 

\begin{figure}[ht]
    \centering
    \includegraphics[width=.61\paperwidth,trim=1.5cm .3cm .6cm .15cm,
    clip]{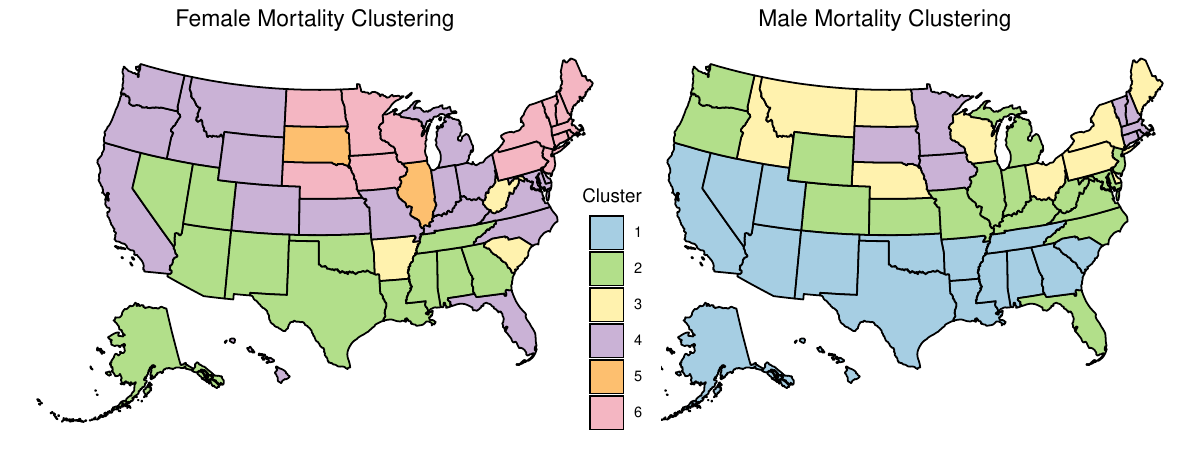}
    \caption{COVID-19 mortality clustering by U.S. state and sex. There are clear spatial effects as well as differences in mortality by sex.  No states have females and males in the same cluster. }
    \label{fig:sidebysideUS}
\end{figure}

\begin{figure}[H]
    \centering
    \includegraphics[width=.4\paperwidth]{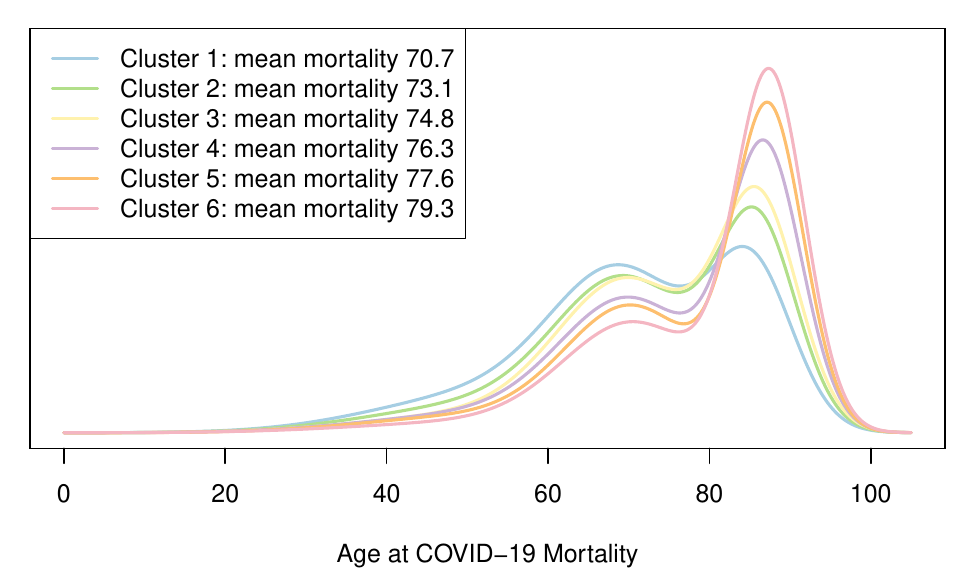}
    \caption{Posterior mean density for each cluster (fixed at the salso estimate), as shown in Figure \ref{fig:sidebysideUS}.  The average age at mortality for each cluster is given in the legend.}
    \label{fig:mortalDistns}
\end{figure}

Re-running the model with the partition fixed at the \textit{salso} point estimate, we show the posterior mean mortality distributions for each of the six clusters in Figure \ref{fig:mortalDistns}.  Together, Figures \ref{fig:sidebysideUS} and \ref{fig:mortalDistns} show that COVID-19 mortality distributions have a substantial geographic effect as well as a strong sex effect. 

While it might be tempting to make broad statements on COVID-19 mortality over geography and sex, we highlight that age demographics of each state surely influence the results.  Additionally, it is well known that in the U.S., females have a longer life expectancy than males \citep{Murphy2024_MortalityUS2023}, so it is not surprising that for every state, the mortality distributions for females have a higher mean than those for males.  

\section{Discussion}\label{sec:disc}

This paper introduces two novel Bayesian mixture models designed for analyzing data available in binned or histogram form, a common scenario in applications involving privacy, summarization, or storage constraints. The proposed models fit a mixture distribution, with a focus on mixtures of normal components, to approximate the underlying continuous density of the population. Additional innovations are the models' ability to scale efficiently with the number of observations, offering a practical alternative to traditional Bayesian nonparametric methods such as Dirichlet process mixture models, which become computationally burdensome for large datasets.  The models and analysis were restricted to the case of histograms arising from univariate data. Generalizations to higher dimensions (2 or 3) might be possible in practice, but the computations would become harder.

We established theoretical support for the first model by proving consistency results under mild regularity conditions, ensuring that the inferred distribution converges to the true data-generating process as the sample size increases. The simulation studies further demonstrate the models' performance in terms of both estimation accuracy and computational speed, comparing favorably with existing approaches. These experiments highlight the method's scalability and robustness across a variety of scenarios.

We conducted a simulation study to demonstrate some desirable properties of the second model.  Additionally, we applied the second model to a COVID-19 mortality dataset composed of multiple histograms. By incorporating a Dirichlet process clustering mechanism, the model allowed for joint analysis of multiple populations, enabling information sharing and providing posterior probabilities for homogeneity across U.S. states and sexes. This application illustrates the practical utility of the approach in complex, real-data settings.

Together, the theoretical guarantees, empirical performance, and real-world applicability make these models a compelling and scalable approach for density estimation and clustering of binned data.

\section*{Acknowledgments}

Fernando Quintana was partially funded by ANID through the FONDECYT Regular project 1260012.

\clearpage
\appendix

\section{Proof of Theorem \ref{thm:mfm-binned-hellinger}}

For the proof, write
\[
B_{n,m} = \left(\frac{m - 1}{M_n},\frac{m}{M_n}\right],\quad m = 1,\ldots,M_n.
\]
For a density \(f\), define the binning map
\[
p_m^{(n)}(f) = \int_{B_{n,m}} f(x)dx,\quad m = 1,\ldots,M_n,
\]
and write
\[
p^{(n)}(f) = \left(p_1^{(n)}(f),\ldots,p_{M_n}^{(n)}(f)\right).
\]
The associated binned density is
\[
T_{M_n}f(x) = M_n \sum_{m = 1}^{M_n} p_m^{(n)}(f)\mathbf{1}_{B_{n,m}}(x).
\]
For a posterior draw \(g\) from model \eqref{eq:singleHistModel}, \(g_b = T_{M_n}g\). We also write
\[
p^{*,n} = p^{(n)}(g^*).
\]
The observed count vector satisfies
\[
\bm v \sim \operatorname{Multinomial}(n,p^{*,n}).
\]
Finally, let
\[
\eta_n = \sqrt{\frac{M_n \log n}{n}} \asymp M_n^{-s} \asymp \left(\frac{\log n}{n}\right)^{s / (2s + 1)}.
\]

We use the following four lemmas.

\begin{lemma}[Hellinger identity for binned densities]
For any two densities \(f\) and \(g\) on \([0,1]\),
\[
h^2(T_{M_n}f,T_{M_n}g) = h^2(p^{(n)}(f),p^{(n)}(g)),
\]
where the Hellinger distance on the right-hand side is the Hellinger distance between probability vectors.
\end{lemma}

\begin{proof}
By construction, \(T_{M_n}f\) and \(T_{M_n}g\) are constant on each bin. Therefore,
\begin{align*}
h^2(T_{M_n}f,T_{M_n}g) &= \frac{1}{2}\sum_{m = 1}^{M_n}\int_{B_{n,m}}\left(\sqrt{M_n p_m^{(n)}(f)} - \sqrt{M_n p_m^{(n)}(g)}\right)^2 dx \\
&= \frac{1}{2}\sum_{m = 1}^{M_n}\frac{1}{M_n}\left(\sqrt{M_n p_m^{(n)}(f)} - \sqrt{M_n p_m^{(n)}(g)}\right)^2 \\
&= \frac{1}{2}\sum_{m = 1}^{M_n}\left(\sqrt{p_m^{(n)}(f)} - \sqrt{p_m^{(n)}(g)}\right)^2 \\
&= h^2(p^{(n)}(f),p^{(n)}(g)).
\end{align*}
\end{proof}

\begin{lemma}[Bias of the reference histogram]
If \(g^*\) is \(\beta\)-Hölder continuous on \([0,1]\), \(s = \beta \wedge 1\), and \(0 < c_- \le g^* \le c_+ < \infty\), then
\[
h(T_{M_n}g^*,g^*) \le C M_n^{-s}
\]
for some constant \(C < \infty\).
\end{lemma}

\begin{proof}
Since \(g^*\) is \(\beta\)-Hölder continuous and \(s = \beta \wedge 1\), there exists \(L < \infty\) such that
\[
|g^*(x) - g^*(y)| \le L |x - y|^s
\]
for every \(x,y \in [0,1]\). If \(x \in B_{n,m}\), then
\begin{align*}
|T_{M_n}g^*(x) - g^*(x)| &= \left|M_n \int_{B_{n,m}} g^*(y)dy - g^*(x)\right| \\
&\le M_n \int_{B_{n,m}} |g^*(y) - g^*(x)|dy \\
&\le L M_n^{-s}.
\end{align*}
Thus,
\[
\|T_{M_n}g^* - g^*\|_\infty \le L M_n^{-s}.
\]
Moreover, since \(T_{M_n}g^*\) is a binwise average of \(g^*\), it also satisfies
\[
c_- \le T_{M_n}g^*(x) \le c_+
\]
for every \(x \in [0,1]\). Hence, for \(a,b \in [c_-,c_+]\),
\[
(\sqrt{a} - \sqrt{b})^2 = \frac{(a - b)^2}{(\sqrt{a} + \sqrt{b})^2} \le \frac{(a - b)^2}{4c_-}.
\]
Therefore,
\[
h^2(T_{M_n}g^*,g^*) \le \frac{1}{8c_-}\int_0^1 (T_{M_n}g^*(x) - g^*(x))^2 dx \le \frac{L^2}{8c_-}M_n^{-2s}.
\]
The result follows.
\end{proof}

\begin{lemma}[Tests in the multinomial model]
Let \(p_0 = p^{*,n}\). There exist tests \(\varphi_n\) and constants \(c > 0\), \(A_0 < \infty\) such that
\[
E_{p_0}\varphi_n \le \exp(-c n \eta_n^2)
\]
and
\[
\sup_{p: h(p,p_0) > A_0\eta_n} E_p(1 - \varphi_n) \le \exp(-c n \eta_n^2)
\]
for all sufficiently large \(n\).
\end{lemma}

\begin{proof}
This is the standard testing bound for multinomial experiments under Hellinger loss. The class of multinomial probability vectors on \(M_n\) categories has bracketing entropy bounded by
\[
H_{[]}(\epsilon,\mathcal P_{M_n},h) \le C M_n \log(C / \epsilon)
\]
for a universal constant \(C < \infty\). Since
\[
n \eta_n^2 = M_n \log n,
\]
and
\[
M_n \log(1 / \eta_n) \le C M_n \log n,
\]
the bracketing-integral condition in the likelihood-ratio testing theorem of Wong and Shen is satisfied. This gives tests with exponentially small type I and type II errors at Hellinger separation of order \(\eta_n\).
\end{proof}

\begin{lemma}[MFM prior thickness for the binned probabilities]
Let $g$ be the random probability density function induced by the mixture of finite mixture model in \eqref{eq:singleHistModel}. 
Under the prior assumptions in Theorem \ref{thm:mfm-binned-hellinger}, there exist constants \(C < \infty\) and \(c > 0\) such that
\[
\Pi\left(K(p^{*,n},p^{(n)}(g)) \le c\eta_n^2,\ V(p^{*,n},p^{(n)}(g)) \le c\eta_n^2\right) \ge \exp(-C n \eta_n^2),
\]
where
\[
K(p,q) = \sum_{m = 1}^{M_n} p_m \log \frac{p_m}{q_m}
\]
and
\[
V(p,q) = \sum_{m = 1}^{M_n} p_m \left(\log \frac{p_m}{q_m}\right)^2.
\]
\end{lemma}

\begin{proof}
Since \(g^*\) is bounded away from zero and infinity, there exist constants \(0 < a_- < a_+ < \infty\) such that
\[
\frac{a_-}{M_n} \le p_m^{*,n} \le \frac{a_+}{M_n}
\]
for all \(m = 1,\ldots,M_n\). We construct a subset of the MFM parameter space on which the induced binned probabilities are close to \(p^{*,n}\).

Set \(K = M_n\), and let
\[
x_m = \frac{m - 1 / 2}{M_n},\quad m = 1,\ldots,M_n
\]
be the bin midpoints. Consider the event \(\mathcal E_n\) on which
\[
|w_m - p_m^{*,n}| \le b\eta_n p_m^{*,n}
\]
for every \(m = 1,\ldots,M_n\),
\[
|\mu_m - x_m| \le \frac{1}{8M_n}
\]
for every \(m = 1,\ldots,M_n\), and
\[
\sigma_m \in \left[\frac{a}{M_n\sqrt{\log n}},\frac{2a}{M_n\sqrt{\log n}}\right]
\]
for every \(m = 1,\ldots,M_n\), where \(a > 0\) is sufficiently small and \(b > 0\) is sufficiently small.

On \(\mathcal E_n\), the \(m\)-th Gaussian component is centered well inside the \(m\)-th bin and has scale of order \(1 / (M_n\sqrt{\log n})\). Therefore its probability of falling outside the \(m\)-th bin is bounded by
\[
\rho_n \le C_1 n^{-C_2}
\]
where \(C_2\) can be made arbitrarily large by choosing \(a\) sufficiently small. In particular, we choose \(a\) so that, for all large \(n\),
\[
\rho_n \le b\eta_n \frac{a_-}{M_n}.
\]
It follows that, on \(\mathcal E_n\),
\[
|p_m^{(n)}(g) - w_m| \le \rho_n
\]
for every \(m = 1,\ldots,M_n\). Therefore,
\[
|p_m^{(n)}(g) - p_m^{*,n}| \le |w_m - p_m^{*,n}| + \rho_n \le 2b\eta_n p_m^{*,n}
\]
for every \(m = 1,\ldots,M_n\). Taking \(b\) sufficiently small, this implies
\[
\left|\frac{p_m^{(n)}(f)}{p_m^{*,n}} - 1\right| \le \frac{1}{2}.
\]
Using the elementary bounds \(|\log(1 + u)| \le C |u|\) and \((\log(1 + u))^2 \le C u^2\) for \(|u| \le 1 / 2\), we obtain
\[
K(p^{*,n},p^{(n)}(g)) \le C_3 \eta_n^2
\]
and
\[
V(p^{*,n},p^{(n)}(g)) \le C_3 \eta_n^2
\]
on \(\mathcal E_n\).

It remains to lower bound \(\Pi(\mathcal E_n)\). By assumption,
\[
\Pi(K = M_n) \ge \exp(-C_K M_n \log M_n).
\]
Conditional on \(K = M_n\), the Dirichlet prior with fixed parameter \(\alpha > 0\) assigns at least
\[
\exp(-C_4 M_n \log n)
\]
mass to the relative neighborhood
\[
|w_m - p_m^{*,n}| \le b\eta_n p_m^{*,n},\quad m = 1,\ldots,M_n.
\]
This follows from Stirling's formula and from the fact that all \(p_m^{*,n}\) are of order \(1 / M_n\).

Next, since the ordered means are the order statistics of i.i.d. draws from a density bounded away from zero on an interval containing \([0,1]\), the prior mass of the event
\[
|\mu_m - x_m| \le \frac{1}{8M_n},\quad m = 1,\ldots,M_n
\]
is bounded below by
\[
\exp(-C_5 M_n \log M_n).
\]
Finally, by the scale-prior assumption,
\begin{align*}
\Pi\left(\sigma_m \in \left[\frac{a}{M_n\sqrt{\log n}},\frac{2a}{M_n\sqrt{\log n}}\right],\ m = 1,\ldots,M_n\right) &\ge \left(c_\sigma \left(\frac{a}{M_n\sqrt{\log n}}\right)^{r_\sigma}\right)^{M_n} \\
&\ge \exp(-C_6 M_n \log n).
\end{align*}
Combining the four lower bounds, and using \(M_n \log M_n \le C M_n \log n\), gives
\[
\Pi(\mathcal E_n) \ge \exp(-C_7 M_n \log n).
\]
Since \(n \eta_n^2 = M_n \log n\), this is
\[
\Pi(\mathcal E_n) \ge \exp(-C_7 n \eta_n^2).
\]
The result follows because \(\mathcal E_n\) is contained in the binned Kullback--Leibler neighborhood displayed at the beginning of the lemma.
\end{proof}

We now prove the theorem. By the testing lemma and the prior-thickness lemma, the posterior contraction theorem for non-i.i.d. triangular-array experiments \citep{GvV2007noniid} yields
\[
\Pi\left(h(p^{(n)}(g),p^{*,n}) > A_1\eta_n \mid \bm v\right) \to 0
\]
in \(P_{g^*}\)-probability, for a sufficiently large constant \(A_1 < \infty\). By the Hellinger identity,
\[
\Pi\left(h(T_{M_n}g,T_{M_n}g^*) > A_1\eta_n \mid \bm v\right) \to 0
\]
in \(P_{g^*}\)-probability.

Finally, by the triangle inequality,
\[
h(T_{M_n}g,g^*) \le h(T_{M_n}g,T_{M_n}g^*) + h(T_{M_n}g^*,g^*).
\]
By the bias lemma,
\[
h(T_{M_n}g^*,g^*) \le C M_n^{-s}.
\]
Therefore, for a sufficiently large constant \(A_2 < \infty\),
\[
\Pi\left(h(T_{M_n}g,g^*) > A_2(\eta_n + M_n^{-s}) \mid \bm v\right) \to 0
\]
in \(P_{g^*}\)-probability. Since \(g_b = T_{M_n}f\), and since
\[
\eta_n + M_n^{-s} \asymp \left(\frac{\log n}{n}\right)^{s / (2s + 1)},
\]
the claim follows.

\end{document}